\documentclass{amsart}
\usepackage{graphicx}
\usepackage{amsmath}
\usepackage{amssymb}
\usepackage{amsthm}

\usepackage{fullpage}

\newtheorem{theorem}{Theorem}[section]
\newtheorem{corollary}[theorem]{Corollary}
\newtheorem{lemma}[theorem]{Lemma}
\newtheorem{proposition}[theorem]{Proposition}
\theoremstyle{definition}
\newtheorem{definition}[theorem]{Definition}
\theoremstyle{remark}

\theoremstyle{definition}

\numberwithin{equation}{section}

\newcommand{\C}{\mathbb C}
\newcommand{\R}{\mathbb R}

\newcommand{\Z}{\mathbb Z}
\newcommand{\E}{{\mathcal E}}
\newcommand{\Ch}{{\mathcal C}}
\newcommand{\eps}{\varepsilon}
\newcommand{\set}[1]{\left\{#1\right\}}

\newcommand{\su}[1]{\mathfrak{su}(#1)}
\newcommand{\ualg}[1]{\mathfrak{u}(#1)}

\newcommand{\matrdpd}[4]{\left( \begin{array}{cc} #1 & #2 \\[0.2cm]  #3 & #4 \end{array} \right)}
\newcommand{\A}{\mathbf A}
\newcommand{\SU}{\mathbf S}
\newcommand{\lie}{\mathfrak{g}}
\newcommand{\LL}{{\mathcal L}}
\newcommand{\MM}{{\mathcal M}}

\usepackage{enumerate}

\begin{document}

\title[Solitons and vortices]{Solitons in gauge theories: existence and dependence on the charge}%
\author{Claudio Bonanno}%
\address{Dipartimento di Matematica, Universit\`a di Pisa, Largo Bruno Pontecorvo n.5, 56127 Pisa, Italy}%
\email{bonanno@dm.unipi.it}%

\begin{abstract}
In this paper we review recent results on the existence of non-topological solitons in classical relativistic nonlinear field theories. We follow the Coleman approach, which is based on the existence of two conservation laws, energy and charge. In particular we show that under mild assumptions on the nonlinear term it is possible to prove the existence of solitons for a set of admissible charges. This set has been studied for the nonlinear Klein-Gordon equation, and in this paper we state new results in this direction for the Klein-Gordon-Maxwell system.
\end{abstract}

\maketitle

\section{Introduction} \label{intro}

In this paper we are interested in the existence of \emph{non-topological soliton solutions} in relativistic classical field theories. The principle governing the existence of these solutions, which contrarily to the topological ones are assumed to vanish at infinity, is due to Coleman \cite{col} and is given by the presence of two conservation laws, energy and charge. In fact solitons are found as critical points of the energy functional restricted to the manifold of function with fixed charge, and in particular as global minimizers of the energy on the manifold. These two properties together imply the main feature of a soliton solution, namely concentration and orbital stability. 

The main field theories are built on the existence of a variational principle, and the corresponding equations are found as Euler-Lagrange equations of an action functional. Moreover it is nowadays well known that the existence of concentrated solutions for field equations is implied by the presence of a nonlinear term in the equation. A first appearance of this principle can be identified in the words of Louis de Broglie 
\begin{quote}
\emph{``Considerations lead me today to believe that the particle must be represented, not by a true point singularity of $u$, but by a very small singular region in space where $u$ would take on a very large value and would obey a non-linear equation, of which the linear equation of Wave Mechanics would be only an approximate form valid outside the singular region. The idea that the equation of propagation of $u$, unlike the classical equation of $\Psi$, is in principle non-linear now strikes me as absolutely essential.''} (\cite[pag. 99]{broglie})
\end{quote}
Finally, for equations of variational nature, the classical Noether's Theorem states that  the existence of conservation laws follows from the existence of group actions which leave the Lagrangian density invariant.

So if we want to study field theories to which to apply the Coleman approach, we have to consider nonlinear field equations with symmetries, and in particular gauge symmetries which are related to the charge invariance.

In this paper we review some recent results in this direction \cite{nlkg-stab,hylo-sol,hylo-vort,bb,bf-vort,bf-stab-g,bon10} for relativistic theories, keeping the discussion to a level as general as possible. Many results are available for non-relativistic theories, for example for the nonlinear Schr\"odinger equations the existence of solitons has been proved in \cite{cazlions}, as well for relativistic theories, see \cite{gss} for a general theory of orbital stability for solutions of Hamiltionian systems. We believe that the advantages of our approach are that the assumptions can be always explicitly verified for a given model.

In Section \ref{sec:abstract} we introduce the notions of matter field and gauge potentials, the definition of solitons and the abstract principle which implies their existence. Then in Sections \ref{sec:nlkg} and \ref{sec:abelian} we show the applications of this principle to relativistic gauge theories, where our aim is to impose on the nonlinear term as less assumptions as possible but sufficient for the existence of solitons. Our approach makes evident the role of the charge, indeed our assumptions are sufficient to prove existence of solitons for a given set of charges. So in general we pay the weakness of the assumptions by the impossibility of proving the existence of solitons for any given charge. The set of admissible charges has been discussed in \cite{bon10} for the nonlinear Klein-Gordon equation of Section \ref{sec:nlkg}. In this paper we prove in Section \ref{sec:nr-abelian} similar new results for the nonlinear Klein-Gordon-Maxwell system, showing in particular that there exist solitons with arbitrarily big electric charge, see Theorem \ref{mio-ab}.

\section{The abstract theory} \label{sec:abstract}

In this section we introduce the abstract framework in which we give precise definitions of  solitary waves, solitons and vortices. We refer to \cite{hylo-sol, bf12, hylo-vort} for more details and discussions.

When talking about gauge theories we mean equations for a couple of fields $(\psi(t,x),\Gamma(t,x))\in X$ with
\[
\psi : \R \times \R^{3} \to \C^{N}\, , \quad N\ge 1
\]
and $\Gamma = (\Gamma_{j})$ for $j=0,1,2,3$ with
\[
\Gamma_{j}: \R \times \R^{3} \to \lie\, , \quad j=0,1,2,3
\]
where $\lie$ is the Lie algebra of a subgroup $G$ of the unitary group $U(N)$. We call $\psi$ the \emph{matter field} and $\Gamma$ the \emph{gauge potentials}. These equations are assumed to be the Euler-Lagrange equations of an action functional
\[
{\mathcal S} = \int_{\R \times \R^{3}}\, \LL(t,x,\psi,\partial_{t}\psi,\nabla \psi, \Gamma, \partial_{t}\Gamma,\nabla \Gamma)\, dtdx
\]
with Lagrangian density $\LL$. In variational systems, the Noether's Theorem implies the existence of a conservation law for any one-parameter Lie group of transformations which leaves invariant the Lagrangian. At least two kind of group actions can be considered\footnote{We denote by square brackets the action of a group member, not to make confusion with the dependence on the space-time variables.}:
\begin{itemize}
\item  \emph{actions on the variables} - this is the case of a group $H=\set{h_{\lambda}}$ which acts on $\R\times \R^{3}$ and induces on $X$ the representation 
\[
H \times X \ni (h_{\lambda},\psi,\Gamma) \mapsto \left(T_{h_{\lambda}}\psi, T_{h_{\lambda}}\Gamma \right) \left( t,x\right) =\Big(\psi(h_{\lambda}[t,x]),\Gamma(h_{\lambda}[t,x]) \Big)\in X\, ;
\]

\item  \emph{gauge actions} - this is the case of a group $H=\set{h_{\lambda}}$ which acts on $\C^{N}\times \lie$ and induces on $X$ the representation  
\[
H \times X \ni (h_{\lambda},\psi,\Gamma) \mapsto \left(T_{h_{\lambda}}\psi, T_{h_{\lambda}}\Gamma \right) \left( t,x\right) =h_{\lambda}\Big[ \psi(t,x), \Gamma(t,x) \Big]\in X\, .
\]
\end{itemize}
For the first kind in this paper we consider Lagrangian densities which are invariant under the action of the Poincar\'e group, giving rise to second order equations in time. Hence we have the following ten conservation laws:
\begin{itemize}
\item  $\E$ - \emph{energy}, the quantity associated to the invariance of the Lagrangian density with respect to time translations $h_{\lambda}(t,x) = (t+\lambda,x)$. We assume that the energy assumes non-negative values;

\item  $\vec{P}$ - \emph{momentum}, the quantity associated to the invariance of the Lagrangian density with respect to space translations, namely the action $h_{\lambda}(t,x) = (t,x+\lambda v)$ for any direction $v\in \R^{3}$;

\item $\vec{L}$ - \emph{angular momentum}, the quantity associated to the invariance of the Lagrangian density with respect to space rotations, namely any one-parameter subgroup $\set{h_{\lambda}}$ of the orthogonal group ${\mathcal O}(3)$, acting as $h_{\lambda}(t,x) = (t,h_{\lambda}(x))$;

\item $\vec{V}$ - \emph{ergocenter velocity}, the quantity associated to the invariance of the Lagrangian density with respect to Lorentz boosts.
\end{itemize}
For the second kind we consider the case of actions of $G<U(N)$ on $\C^{N}$ by standard matrix representation, and on $\lie$ by a translate of the adjoint representation. However, looking at $\C^{N}\times \lie$ as the fiber of a trivial bundle with $\R\times \R^{3}$ as base space, the main feature is whether the gauge action depends or not on the point of the base space. This is discussed in Section \ref{sec:gvl} and gives rise to the notions of \emph{global} and \emph{local gauge actions}. In both cases, associated to the gauge actions we have $k$ $(=\dim G)$ conservation laws, and we call \emph{hylomorphic charges} the respective invariant quantities $\Ch$.

To give definitions of \emph{solitary waves} and \emph{solitons}, we need to introduce a dynamical point of view. We can think of the solutions $(\psi(t,x),\Gamma(t,x))\in X$ of our field equations as orbits of a dynamical system defined by a time evolution map $U:\R\times Y \to Y$ defined for all $t\in \R$, where $Y$ is the phase space of the system given by the couples $(\Psi, \Omega)$, with $\Psi = (\psi, \partial_{t} \psi)$ and $\Omega = (\Gamma, \partial_{t} \Gamma)$. The form of $Y$ comes from the fact that the equations are of second order in time. So if $(\Psi_0, \Omega_{0})\in Y$ are the initial conditions of our equations, the evolution of the system is described by 
\begin{equation}\label{evol}
(\Psi(t,x),\Omega(t,x))=U\Big(t,(\Psi_0, \Omega_{0})\Big)
\end{equation}

We assume that for all $(\psi,\Gamma)\in X$ it holds $\psi \in L^{2}(\R^{3},\C^{N})$. This implies that for the orbits of our system we can define the \emph{barycenter} of the matter field as
\[
\vec{q}_{_{\psi }}(t)=\frac{\int_{\R^{3}} x\left| \psi (t, x) \right|_{\C^{N}}^{2}dx}{\int_{\R^{3}} \left| \psi (t, x) \right|_{\C^{N}}^{2}dx}
\]
The term \emph{solitary wave} is usually used for solutions of field equations for which the energy of the matter field is localized. Using the notion of barycenter, we give a formal definition of solitary wave. 
\begin{definition} \label{solw}
A state $(\Psi_{0},\Omega_{0})$ is called \emph{solitary wave} if for any $\varepsilon >0$ there exists a radius $R>0$ such that for all $t\in \R$
\[
\int_{\R^{3}} \left| \psi (t, x) \right|_{\C^{N}}^{2} dx-\int_{B_{R}( \vec{q}_{_{\psi }}(t))}\left| \psi (t, x) \right|_{\C^{N}}^{2} dx<\varepsilon
\]
where $\Psi(t,x) = (\psi(t,x),\partial_{t}\psi(t,x))$ and $(\Psi(t,x),\Omega(t,x))=U\Big(t,(\Psi_{0},\Omega_{0})\Big)$. Moreover $B_{R}( \vec{q}_{_{\psi }}(t))$ denotes the ball in $\R^{3}$ of radius $R$ and center $\vec{q}_{_{\Psi }}(t)$.
\end{definition}
\begin{definition} \label{solw-v}
A state $(\Psi_{0},\Omega_{0})$ is called \emph{vortex} if it is a solitary wave and the angular momentum $\vec{L}(\psi(t,x),\Gamma(t,x))$ does not vanish.
\end{definition}
The \emph{solitons} are solitary waves which are \emph{orbitally stable}. 
\begin{definition} \label{ds}
A state $(\Psi_{0},\Omega_{0})$ is called \emph{soliton} if it is a solitary wave and the matter field is orbitally stable, that is there exists a finite dimensional manifold $\MM$ with $\Psi_{0}\in \MM$ such that:
\begin{itemize}
\item $\MM$ is $U$-\emph{invariant}, that is for any $(\Phi_{0},\tilde \Omega_{0})$ with $\Phi_{0}\in \MM$, it holds $\Phi(t,x)\in \MM$ for all $t\in \R$, where $(\Phi(t,x),\tilde \Omega(t,x))$ is the evolution of $(\Phi_{0},\tilde \Omega_{0})$ as defined in \eqref{evol};
\item $\MM$ is $U$-\emph{stable}, that is for any $\eps >0$ there exists $\delta >0$ such that if $d(\Phi_{0},\MM)< \delta$, for some $(\Phi_{0},\tilde \Omega_{0})$, then $d(\Phi(t,x),\MM)< \eps$ for all $t\in \R$, where $(\Phi(t,x),\tilde \Omega(t,x))$ is defined as above and $d$ is a distance on the space of matter fields.
\end{itemize}
\end{definition}

In our approach the existence of conservation laws is fundamental to obtain solitons. Indeed given the set of functions with fixed charge
\[
\Sigma_{\sigma} = \set{ (\psi,\Gamma) \in X\, : \, \Ch(\psi,\Gamma) = \sigma}
\]
we obtain solitons first proving that the energy $\E$ has minimum on $\Sigma_{\sigma}$,  then showing that the set $\MM_{\sigma}$ of minimizers is made of solitary waves and that $\MM_{\sigma}$ is a finite dimensional manifold which is $U$-invariant and $U$-stable. For an approach to stability of solitary waves in Hamiltonian PDEs see \cite{gss}. The main difference with our approach is that we give sufficient conditions for stability which depend only on the energy functional, whereas to check the sufficient conditions given in \cite{gss} one needs to have more information about the solution.

\section{Global gauge theory: the nonlinear Klein-Gordon equation} \label{sec:nlkg}

In this section we review the results for the case of \emph{global gauge actions}, namely for the case of a group $H=\set{h_{\lambda}}_{\lambda\in \R}$ acting on $\C^{N}\times \lie$, where the elements $h_{\lambda}$ do not depend on the variables $(t,x)$. In particular we consider the simplest dynamical system which is generated by a Lagrangian density which is invariant for the action of the Poincar\'e group.

Let $\Gamma\equiv 0$ and the matter field $\psi(t,x) \in H^{1}(\R \times \R^{3},\C) = X$. We consider the Lagrangian density
\begin{equation}\label{lag-nlkg}
\LL(\psi,\partial_{t}\psi) = \frac 12 |\partial_{t}\psi|^{2} - \frac 12 |\nabla \psi|^{2} - W(|\psi|)
\end{equation}
for a $C^{2}$ function $W:\R^{+}\to \R$. The Euler-Lagrange equation of $\LL$ is the \emph{nonlinear Klein-Gordon equation}
\begin{equation}\label{nlkg}
\partial_{t}^{2}\psi - \triangle \psi + W'(|\psi|)\, \frac{\psi}{|\psi|} = 0
\tag{NLKG}
\end{equation}
The Lagrangian density \eqref{lag-nlkg} is invariant for the action of the Poincar\'e group on $(t,x)$, which implies the existence of ten conservation laws: energy, momentum, angular momentum and ergocenter velocity. In particular energy takes the form
\begin{equation}\label{en-nlkg}
\E(\psi,\partial_{t}\psi) = \int_{\R^{3}}\, \Big( \frac 12 |\partial_{t}\psi|^{2} + \frac 12 |\nabla \psi|^{2} + W(|\psi|) \Big)\, dx
\end{equation}
and the angular momentum is given by
\[
\vec{L}(\psi,\partial_{t}\psi) = \Re \int_{\R^{3}}\, \Big( \vec{x} \times \nabla \psi \Big)\, \overline{\partial_{t}\psi}\, dx
\]
Moreover, since the Lagrangian $\LL$ only depends on the modulus of $\psi$ and $\partial_{t}\psi$, it is invariant also for the action of the one-dimensional global gauge group $U(1) \cong S^{1} = \set{e^{i\lambda}}_{\lambda\in \R}$ which is given by
\begin{equation}\label{ga-nlkg}
S^{1}\times X \ni (e^{i\lambda}, \psi) \mapsto e^{i\lambda}\psi(t,x) \in X
\end{equation}
By Noether's theorem we obtain one more conservation law, which we call \emph{hylomorphic charge} $\Ch$, and which is given by
\[
\Ch(\psi,\psi_{t}) = \Im \int_{\R^{3}}\, \overline{\psi}\, \partial_{t}\psi \, dx
\]
For equation \eqref{nlkg}, the easiest way to produce a solitary wave solution is to look for solutions of the form
\begin{equation}\label{ansatz-nlkg}
\psi(t,x) = u(x)\, e^{-i\omega t}
\end{equation}
for $u(x):\R^{3}\to \R^{+}$ in $H^{1}(\R^{3})$ and $\omega \in \R$. Notice that these functions are an orbit of $u(x)$ for the action of the gauge group $S^{1}$ with $\lambda = -\omega t$. A function of the form \eqref{ansatz-nlkg} is a solution of \eqref{nlkg} if it satisfies
\begin{equation}\label{nlkg-s}
- \triangle u -\omega^{2} u +W'(u) =0 
\end{equation}
It is immediate to verify that functions of the form \eqref{ansatz-nlkg} satisfy Definition \ref{solw} with $\vec{q}_{_{\psi}}(t) = \vec{q}_{_{\psi}}(0)$ for all $t\in \R$. We can introduce the space of solitary waves of form \eqref{ansatz-nlkg}
\[
X_{S}:= \set{(u,\omega) \in H^{1}(\R^{3},\R^{+}) \times \R}
\]
which is embedded into $X$ by
\[
X_{S} \ni (u,\omega) \mapsto u(x)e^{-i\omega t} \in X
\]
Moreover we consider the energy and charge functionals on $X_{S}$. We get
\[
E(u,\omega) := \E|_{X_{S}} = \int_{\R^{3}}\, \Big( \frac 12 |\nabla u|^{2} + W(u) + \frac 12 \omega^{2} u^{2} \Big)\, dx
\]
\[
C(u,\omega) := \Ch|_{X_{S}} = - \int_{\R^{3}}\, \omega u^{2} \, dx
\]
Moreover notice that $\vec{L}|_{X_{S}}\equiv 0$, hence $X_{S}$ does not contain vortices.

We now sketch the steps to prove the existence of solitons for \eqref{nlkg}. First of all we need to show that there are solitary waves, namely couples $(u,\omega) \in X_{S}$ solutions of \eqref{nlkg-s}. The first result of existence of solutions for equations like \eqref{nlkg-s} in a general form dates back to the classical paper by Berestycki and Lions \cite{bl}, see also \cite{cgm}. Here we use the following simple remark
\begin{proposition}[\cite{nlkg-stab}] \label{critvinc}
A couple $(u,\omega) \in X_{S}$ is a solution of equation \eqref{nlkg-s} if and only if $(u,\omega)$ is a critical point of the energy $E(u,\omega)$ constrained to the manifold
\[
\Sigma^{S}_{\sigma}:= \set{(u,\omega) \in X_{S}\, : \, C(u,\omega) = \sigma}
\]
\end{proposition}
We are then reduced to prove the existence of critical points of $E$ constrained to $\Sigma^{S}_{\sigma}$ for some $\sigma\in \R$. The easiest way to prove the existence of such critical point is to show that $E$, which is a differentiable functional, has a point of minimum on $\Sigma^{S}_{\sigma}$. It turns out that points of minima are relevant also for the second part of the existence of a soliton for \eqref{nlkg}, namely the proof that the found solitary wave is orbitally stable.

Let $(u_{0}, \omega_{0})\in X_{S}$ be a minimizer of $E$ on $\Sigma^{S}_{\sigma}$, where $\sigma=C(u_{0}, \omega_{0})$. Then, since the energy $E$ is invariant under the action of the Poincar\'e group and of the gauge group $S^{1}$, it follows that we actually have a finite dimensional manifold of minimizers for $E$, and henceforth for $\E$, given by
\[
\MM(u_{0},\omega_{0}) = \set{ \psi(t,x) = u_{0}(x+a) e^{i(-\omega_{0}t +\theta)}\, :\, a\in \R^{3},\, \theta \in \R}
\]
Notice that for all $\psi \in \MM(u_{0},\omega_{0})$ we have $\Ch(\psi)= C(u_{0}, \omega_{0}) = \sigma$.

We say that $(u_{0}, \omega_{0})$ is an \emph{isolated} point of minimum for $E$ if for any other minimizer $(u_{1},\omega_{1}) \in \Sigma^{S}_{\sigma}$, with $u_{1}(x)e^{-i\omega_{1}t}\not\in \MM(u_{0},\omega_{0})$, it holds
\[
\MM(u_{0},\omega_{0}) \cap \MM(u_{1},\omega_{1}) = \emptyset\, .
\]

\begin{theorem}[\cite{nlkg-stab}] \label{orbstab-min}
If $(u_{0}, \omega_{0})$ is an isolated point of local minimum for $E$ constrained to $\Sigma^{S}_{\sigma}$, then $\MM(u_{0},\omega_{0})$ is a stable manifold for the flow associated to the nonlinear Klein-Gordon equation. In particular $\psi(t,x) = u_{0}(x) e^{-i\omega_{0}t}$ is a soliton solution to \eqref{nlkg}.
\end{theorem}

Hence, putting together Proposition \ref{critvinc} and Theorem \ref{orbstab-min} we need to show the existence of an isolated point of local minimum for $E$ constrained to $\Sigma^{S}_{\sigma}$ for some $\sigma$. It follows from \cite{nlkg-stab} and \cite{bon10} that it is possible to study the existence of such point of minimum depending on the value of $\sigma$. 

We now introduce the assumptions on the nonlinear term $W:\R^{+}\to \R$. We assume that $W$ is of class $C^{2}$ and of the form 
\[
W(s) = \frac 12 m^{2} s^{2} + R(s)
\]
such that 
\begin{itemize}
\item[(W0)] $m>0$ and $R(0)=R'(0)=R''(0)=0$;
\item[(W1)] $R(s) \ge - \frac 12 m^{2} s^{2}$ for all $s\in \R^{+}$;
\item[(W2)] there exists $s_{0}>0$ such that $R(s_{0})<0$;
\item[(W3)] there exist positive constants $c_{1},c_{2}$ such that
\[
|R''(s)|\le c_{1}s^{p-2} + c_{2}s^{q-2}
\]
for all $s\in \R^{+}$ and some $2<p,q< 6$.
\end{itemize}
We briefly comment on these assumptions. (W0) simply implies $W''(0)=m^{2}\not= 0$, which can be interpreted as a non-vanishing condition for the ``mass'' of the matter field $\psi$. (W1) implies that $W(s)\ge 0$, so that the energy \eqref{en-nlkg} is non-negative. (W2) and (W3) are standard assumptions in the variational approach to elliptic equations. In particular (W2) is fundamental for the existence of solitary waves, as was already observed in \cite{bl}. Finally (W3) says that $W$ is sub-critical with respect to the Sobolev embedding. This assumption can be weakened as discussed for example in \cite{nlkg-stab}.

Let introduce the notation
\[
X_{S}^{-} := \set{ u\in H^{1}(\R^{3},\R^{+})\, :\, J(u):= \int_{\R^{3}}\, \Big( \frac 12 |\nabla u|^{2} + R(u) \Big)\, dx <0}
\]
and 
\[
\sigma_{g} := \inf_{u\in X_{S}^{-}} \Big( m\|u\|_{L^{2}}^{2} - \|u\|_{L^{2}} \sqrt{2\, |J(u)|} \Big) \ge 0
\]
Putting together the results from \cite{nlkg-stab} and \cite{bon10} we state the following
\begin{theorem}[\cite{nlkg-stab},\cite{bon10}] \label{beh-carica}
Under assumptions (W0)-(W3) on the nonlinear term $W$, we have
\begin{enumerate}[(i)]
\item if $|\sigma|> \sigma_{g}$ then $E(u,\omega)$ admits a point of global minimum on $\Sigma_{\sigma}^{S}$;
\item if $|\sigma|\le \sigma_{g}$ then $\inf_{\Sigma_{\sigma}^{S}} E(u,\omega)$ is not attained;
\item if there exist $\alpha>0$ and $\eps\in (0,\frac 43)$ such that $R(s)<0$ for $s\in (0,\alpha)$ and $\limsup_{s\to 0^{+}}\, \frac{|R(s)|}{s^{2+\eps}} >0$, then $\sigma_{g}=0$;
\item if $\sigma_{g}=0$, then there exists $\alpha>0$ such that $R(s)<0$ for $s\in (0,\alpha)$ and $\limsup_{s\to 0^{+}}\, \frac{|R(s)|}{s^{2+\frac 43}} >0$;
\item if $\sigma_{g}>0$ there exists $\sigma_{b} < \sigma_{g}$ such that if $\sigma \in (\sigma_{b},\sigma_{g}]$ then $E(u,\omega)$ admits a point of local minimum on $\Sigma_{\sigma}^{S}$;
\item if $\sigma_{g}>0$ and there exists $s_{1}>0$ such that $R(s_{1}) = -\frac 12 m^{2} s_{1}^{2}$, then $\sigma_{b}=0$.
\end{enumerate}
\end{theorem}
It follows from Theorem \ref{beh-carica} that we have information about the existence of a soliton of charge $\sigma$ according to the behaviour of the nonlinear term $W(s)$. See \cite{garrisi} for a result on a system of Klein-Gordon equations using this approach. 

We now consider the existence of vortices for the nonlinear Klein-Gordon equation. We refer to \cite{hylo-vort} for more details (see also \cite{badrol}). As stated above, functions of the form \eqref{ansatz-nlkg} have vanishing angular momentum. Hence we have to change the ansatz. For $x\in \R^{3}$ let us write $x=(y,z) \in \R^{2}\times \R$, and consider functions of the form  
\begin{equation}\label{ansatz-nlkg-vort}
\psi(t,x) = u(x)\, e^{i( \ell \theta(y)-\omega t)}
\end{equation}
for $u(x):\R^{3}\to \R^{+}$ in $H^{1}(\R^{3})$, $\omega \in \R$, $\ell \in \Z$ and
\begin{equation}\label{angolo}
\theta(y) := \Im \, \log(y_{1}+iy_{2}) \in \R/2\pi\Z
\end{equation}
is the angular variable in the $(y_{1},y_{2})$-plane. Letting $r:=\sqrt{y_{1}^{2}+y_{2}^{2}}$, by definition $\theta$ satisfies
\[
\triangle \theta =0\, , \quad \nabla \theta = \Big( -\frac{y_{2}}{r^{2}},\, \frac{y_{1}}{r^{2}},\, 0\Big)\, , \quad |\nabla \theta| = \frac 1r
\]
It follows that a function $\psi$ of the form \eqref{ansatz-nlkg-vort} is a solution of \eqref{nlkg} if the triple $(u,\omega,\ell)$ is a solution of
\begin{equation}\label{nlkg-vort}
-\triangle u + \Big( \frac{\ell^{2}}{r^{2}} - \omega^{2}\Big) u + W'(u)=0
\end{equation}
Computing the energy, charge and angular momentum on functions of the form \eqref{ansatz-nlkg-vort} we find
\[
E(u,\omega,\ell) := \int_{\R^{3}}\, \Big( \frac 12 |\nabla u|^{2} + W(u) + \frac 12 \Big(\frac{\ell^{2}}{r^{2}} + \omega^{2}\Big) u^{2} \Big)\, dx
\]
\[
C(u,\omega,\ell) :=  - \int_{\R^{3}}\, \omega u^{2} \, dx
\]
\[
\vec{L}(u,\omega,\ell) :=  \Big( 0,\, 0\, , - \int_{\R^{3}}\, \ell \omega u^{2} \, dx\Big) = \Big( 0,\, 0\, , \ell C(u,\omega,\ell) \Big)
\]
Hence if we find a solution to \eqref{nlkg-vort} with $\ell\not=0$ and non-vanishing charge, then we have a vortex solution to \eqref{nlkg}. This is accomplished as for solitary waves by first noticing that the analogous of Proposition \ref{critvinc} holds. Namely,
\begin{proposition}[\cite{hylo-vort}] \label{critvinc-v}
Let $\ell \in \Z\setminus \set{0}$ be fixed. The triple $(u,\omega,\ell)$ is a solution of equation \eqref{nlkg-vort} if and only if $(u,\omega)$ is a critical point of the energy $E(u,\omega,\ell)$ constrained to the manifold $\Sigma^{S}_{\sigma}$.
\end{proposition}
Again the easiest way to find constrained critical points for $E(u,\omega,\ell)$ is to look for minimizers on $\Sigma^{S}_{\sigma}$. A weaker version of Theorem \ref{beh-carica} holds:
\begin{theorem}[\cite{hylo-vort}] \label{beh-vort}
Under assumptions (W0)-(W3) and for any fixed $\ell \in \Z\setminus \set{0}$, there exists $\sigma_{0} >0$ such that if $|\sigma|>\sigma_{0}$ the energy $E(u,\omega,\ell)$ admits a point of global minimum on $\Sigma^{S}_{\sigma}$. In particular the nonlinear Klein-Gordon equation admits a vortex solution with finite energy, charge $\sigma$ and angular momentum $\ell \sigma$.
\end{theorem}
The orbital stability of these vortex solutions is open at this moment. However in \cite{hylo-vort} we give some analytical and numerical results that suggest that these solutions are unstable.

\section{Global vs local gauge theories} \label{sec:gvl}

In the last section we have considered the Euler-Lagrange equations related to a simple Lagrangian density $\LL$ depending only on the matter field. The Lagrangian $\LL$ was invariant under the action of the global gauge group $S^{1}$. Now we examine how a Lagrangian density has to change if we want to consider the action of a local gauge group. For this section we refer to \cite{yang}.

Let us consider the Lagrangian density \eqref{lag-nlkg} with $\psi(t,x) \in \C^{N}$ and a group $G < U(N)$ with Lie algebra $\lie$. Let us consider the gauge action on $\C^{N}$ of $G$-valued functions $g(t,x)\in G$ defined in $\R\times \R^{3}$. So for each $(t,x)\in \R\times \R^{3}$, we consider the action
\[
G \times \C^{N} \ni (g(t,x),\psi(t,x)) \mapsto g(t,x)\Big[ \psi(t,x)\Big]\in \C^{N}\, .
\]
Let us see how $\LL$ changes when evaluated on $g(t,x)[ \psi(t,x)]$. The last term is unchanged, $W(|g[\psi]|) = W(|\psi|)$, since $g(t,x)\in G<U(N)$ for each $(t,x)$. Instead the terms containing the derivatives of $\psi$ become
\[
\Big| \partial_{j}\, \Big(g(t,x)[\psi(t,x)]\Big) \Big| = \Big| \Big(\partial_{j} g(t,x)\Big) \psi(t,x) + g(t,x)\, \partial_{j} \psi(t,x) \Big|
\]
for $j=0,1,2,3$, where\footnote{The signs come from the choice of a metric on $\R\times \R^{3}$ (see \cite{bb}).} $\partial_{0}=-\partial_{t}$ and $\nabla = (\partial_{1},\partial_{2},\partial_{3})$. One way to keep invariance also of the terms with derivatives is to substitute $\{\partial_{j}\}_{j=0,1,2,3}$ with the \emph{covariant derivatives} 
\begin{equation}\label{deriv-cov}
D_{j} := \partial_{j} +q\, \Gamma_{j}(t,x)\, , \qquad j=0,1,2,3
\end{equation}
where $q>0$ is a real parameter, which is the strength of the action of $\Gamma$ on the matter field, and $\Gamma=(\Gamma_{j})$ are the gauge potentials, that is $\lie$-valued functions. The covariant derivatives have been introduced in differential geometry to differentiate functions defined on manifolds along tangent vectors. In this approach the gauge potentials are called \emph{connection}. We refer the reader to \cite{jost}.

Hence let
\begin{equation}\label{lag-matter}
\LL_{0}(\psi,\partial_{t}\psi,\nabla \psi) := \frac 12 |D_{0}\psi|^{2} - \frac 12 \sum_{j=1}^{3} |D_{j} \psi|^{2} - W(|\psi|)
\end{equation}
Denoting $\tilde \Gamma_{j}(t,x) = g(t,x)[\Gamma_{j}(t,x)]$, we have
\[
\Big| \Big(\partial_{j} +q\, \tilde \Gamma_{j}(t,x)\Big) \Big(g(t,x)[\psi(t,x)]\Big) \Big| =
\]
\[
= \Big| \Big(\partial_{j} g(t,x)\Big) \psi(t,x) + g(t,x)\, \partial_{j} \psi(t,x) + q\, \tilde \Gamma_{j}(t,x) g(t,x)[\psi(t,x)]\Big|=
\]
\[
= \Big| g(t,x) \Big[ \partial_{j} \psi(t,x) +\Big( q\, g^{-1}(t,x) \tilde \Gamma_{j}(t,x) g(t,x) + g^{-1}(t,x) \partial_{j} g(t,x) \Big) \psi(t,x) \Big] \Big| =
\]
\[
= \Big| \partial_{j} \psi(t,x) +\Big( q\, g^{-1}(t,x) \tilde \Gamma_{j}(t,x) g(t,x) + g^{-1}(t,x) \partial_{j} g(t,x) \Big) \psi(t,x) \Big|
\]
where in the last equality we have used again that $g(t,x)\in G<U(N)$ for each $(t,x)$. Finally, letting
\[
q\, g^{-1}(t,x) \tilde \Gamma_{j}(t,x) g(t,x) + g^{-1}(t,x) \partial_{j} g(t,x) = q\, \Gamma_{j}(t,x)
\]
it follows that
\[
\Big| \Big(\partial_{j} +q\, \tilde \Gamma_{j}(t,x)\Big) \Big(g(t,x)[\psi(t,x)]\Big) \Big| = \Big| \Big(\partial_{j} +q\, \tilde \Gamma_{j}(t,x)\Big) \psi(t,x) \Big|
\]
Hence the Lagrangian density \eqref{lag-matter} is invariant for the action of a local gauge group $G<U(N)$ if we define the gauge action of $G$ on the gauge potentials $\Gamma$ by
\begin{equation}\label{azione-gauge}
g(t,x)[\Gamma_{j}(t,x)] = g(t,x) \Gamma_{j}(t,x) g^{-1}(t,x) - \frac 1q  \Big( \partial_{j} g(t,x) \Big) g^{-1}(t,x)
\end{equation}

Finally typically one wants to study systems in which the gauge potentials are not an external action on the matter field, but are instead part of the system. In this case one needs to add another term to the Lagrangian density to drive the evolution of $\Gamma$. It turns out that one of the simplest terms which is invariant under the action \eqref{azione-gauge} of the gauge group is given by
\begin{equation}\label{lag-gauge}
\LL_{1}(\Gamma,\partial_{t}\Gamma,\nabla \Gamma) := \frac12 \, \sum_{j=1}^{3}\, \| F_{0j} \|^{2} - \frac14 \, \sum_{k,j=1}^{3}\, \| F_{kj} \|^{2}
\end{equation}
where $F=(F_{kj})$ is the \emph{strength of the gauge field}, or the \emph{curvature} of the connection $\Gamma$ in differential geometry, with
\begin{equation}\label{curvature}
F_{kj} := \partial_{k}\Gamma_{j} - \partial_{j}\Gamma_{k} + q\, [\Gamma_{k},\Gamma_{j}] \in \lie
\end{equation}
and $[\cdot,\cdot]$ is the standard commutator, and finally $\|U\|^{2}:= \text{trace}(U^{*}U)$ is the Hilbert norm on $\lie$.

In the next sections we study local gauge theories with $N=1$ and $N=2$ using the Lagrangian density $\LL = \LL_{0}+\LL_{1}$.

\section{Local gauge theory: the Abelian case} \label{sec:abelian}

We first consider the case $N=1$ and $G=U(1)$, so that $\psi(t,x) \in \C$ and $\Gamma=(\Gamma_{j})$ with $\Gamma_{j}(t,x) \in \lie = \ualg{1} = i \R$. This is called the Abelian case because the gauge group $G$ is Abelian.

The system of equations that we obtain is called \emph{Klein-Gordon-Maxwell} system, since as we show below, it can be interpreted as the system for a charged particle interacting with itself through the nonlinear term $W$ and with an electromagnetic field with potentials $\Gamma$. For this reason we use the notation of $\Gamma$ as a four-vector with components
\[
\Gamma = (-i\varphi, i\A) \quad \mbox{where} \quad \A=( A_{1}, A_{2}, A_{3})
\]
The covariant derivatives \eqref{deriv-cov} then take the form
\[
D_{0} \psi := \left( - \partial_t - i q \varphi \right) \psi\, , \qquad D_{j} \psi :=  \left( \partial_{j} + i q A_{j} \right) \psi, \quad  j=1,2,3
\]
Using this notation we rewrite $\LL_{0}$ in \eqref{lag-matter} as follows:
\begin{equation} \label{elle-0-p-abel}
\LL_{0}(\psi, \partial_{t} \psi, \nabla \psi) = \frac 1 2\, \left| \partial_{t} \psi + i\, q\, \varphi \psi \right|^{2} - \frac 1 2\, \left| \nabla \psi + i\, q\, \A \psi \right|^{2} -W(|\psi|)
\end{equation}
To write $\LL_{1}$ in \eqref{lag-gauge} we first compute the components $F_{kj}$ defined in \eqref{curvature}, which in this case are complex numbers given by
$$
F_{0j} = - i \partial_{t} A_{j} + i \partial_{j} \varphi\, , \qquad j=1,2,3
$$
$$
F_{kj} = i \partial_{k} A_{j} - i \partial_{j} A_{k}\, , \qquad k,j=1,2,3
$$
and $\| F_{kj} \|^{2} = |F_{kj}|^{2}$. It follows that
$$
\sum_{j=1}^{3}\, \| F_{0j}\|^{2} = \sum_{j=1}^{3}\, \left( \partial_{t} A_{j} - \partial_{j} \varphi \right)^{2} = \left| \partial_{t} \A - \nabla \varphi \right|^{2}
$$
$$
\sum_{k,j=1}^{3}\, \|F_{ij}\|^{2} = 2\ \left( |F_{12}|^{2} + |F_{23}|^{2} + |F_{31}|^{2} \right)= 2\ \left| \nabla \times \A \right|^{2}
$$
and
\begin{equation} \label{elle-1-p-abel}
\LL_{1}(\Gamma, \partial_{t}\Gamma,\nabla \Gamma) =  \frac 1 2\,  \left| \partial_{t} \A - \nabla \varphi \right|^{2} - \frac 1 2\, \left| \nabla \times \A \right|^{2}
\end{equation}
In this case the gauge action on the matter field $\psi$ is the same as in the Klein-Gordon equation and is given by \eqref{ga-nlkg}, and the action \eqref{azione-gauge} becomes
\begin{equation}\label{azione-abel}
e^{i\lambda(t,x)}[\Gamma_{j}(t,x)] = \Gamma_{j}(t,x) - \frac iq \partial_{j} \lambda(t,x) 
\end{equation}
Using \eqref{azione-abel} it is easy to verify that $e^{i\lambda(t,x)}[F_{kj}(t,x)] = F_{kj}(t,x)$ for each $k,j=0,1,2,3$. 

To obtain the Klein-Gordon-Maxwell system of equations we make the variations of ${\mathcal S}=\int (\LL_{0}+\LL_{1})$ with respect to $\psi$, $\varphi$ and $\A$, and obtain
\begin{align}
& D^{2}_{0}\, \psi - \sum_{j=1}^{3}\, D^{2}_{j}\, \psi + W'(\psi)=0 \label{prima-kgm} \\
& \nabla \cdot \left( \partial_{t} \A - \nabla \varphi \right) + q\, \Re(i\, \psi\, \partial_t \bar \psi) + q^2\, |\psi|^2 \, \varphi = 0  \label{seconda-kgm} \\
& \partial_{t}  \left( \partial_{t} \A - \nabla \varphi \right)  + \nabla \times (\nabla \times \A) + q\, \Re (i\, \psi \, \nabla \bar \psi) + q^2\, |\psi |^2 \A= 0  \label{terza-kgm}
\end{align}
and we look for solutions 
\[
(\psi,\varphi,\A) \in X= H^{1}(\R\times\R^{3},\C) \times \dot H^{1}(\R\times\R^{3},\R)\times (\dot H^{1}(\R\times\R^{3},\R))^{3}\, .
\]

A useful approach to equations (\ref{prima-kgm})-(\ref{terza-kgm}) is to look for solutions $\psi(t,x) \in \C$ written in polar form, that is
\begin{equation} \label{polar-1}
\psi(t,x) = u(t,x) \, e^{i \, S(t,x)}, \qquad u\in \R^{+}, \ S\in \R/2\pi \Z
\end{equation}
Using notation (\ref{polar-1}), equation (\ref{prima-kgm}) splits in the equations
\begin{align}
& \partial_{t}^{2} u -\triangle u + \left[ |\nabla S + q\A|^{2} - \left( \partial_{t} S +q\varphi \right)^{2} \right] \, u + W'(u) =0 \label{prima-kgm-p1} \\
& \partial_{t} \left[ \left( \partial_{t}S +q\varphi \right) u^{2} \right] - \nabla \cdot \left[ \left( \nabla S + q\A \right) u^{2} \right] =0 \label{prima-kgm-p2}
\end{align}
and (\ref{seconda-kgm}) and (\ref{terza-kgm}) become
\begin{align}
& \nabla \cdot \left( \partial_{t} \A - \nabla \varphi \right) + q\, (\partial_{t} S + q\varphi) \, u^{2}=0 \label{seconda-kgm-p} \\
& \partial_{t}  \left( \partial_{t} \A - \nabla \varphi \right)  + \nabla \times (\nabla \times \A) + q\, (\nabla S + q \A)\, u^{2} =0  \label{terza-kgm-p}
\end{align}
If we make the identifications
\[
{\mathbf E} = - \partial_{t} \A + \nabla \varphi\, , \quad {\mathbf H} = \nabla \times \A
\]
with $\mathbf E$ the electric field and $\mathbf H$ the magnetic field, and
\[
\rho = q \left( \partial_{t}S +q\varphi \right) u^{2}\, , \quad \mathbf{j} = - q\left( \nabla S + q\A \right) u^{2}
\]
with $\rho$ the electric charge density and $\mathbf{j}$ the electric current density, it follows that equation \eqref{prima-kgm-p2} is the continuity equation for the electric charge density, equation \eqref{seconda-kgm-p} is the Gauss equation and equation \eqref{terza-kgm-p} is the Ampere equation. Moreover the Faraday equation and the null-divergence equation for the magnetic field are automatically satisfied. Hence \eqref{prima-kgm-p1}-\eqref{terza-kgm-p} is the called the Klein-Gordon-Maxwell system.

The Lagrangian density $\LL = \LL_{0}+ \LL_{1}$ given by \eqref{elle-0-p-abel} and \eqref{elle-1-p-abel} is invariant for the action of the Poincar\'e group, hence we obtain the ten conservation laws given by energy
\[
\E = \frac 12\ \int_{\R^3} \, \Big[ (\partial_{t} u)^{2} + |\nabla u|^{2} + \frac{\rho^{2} + |{\mathbf j}|^{2}}{q^{2} u^{2}} + 2W(u) + |\partial_{t} \A - \nabla \varphi|^{2} + |\nabla \times \A|^{2}  \Big]\, dx 
\]
momentum $\vec{P}$, angular momentum
\[
\vec{L} = \int_{\R^{3}}\, {\mathbf x} \times \left[ \partial_{t}u \, \nabla u - \frac{\rho\, \nabla S}{q^{2} u^{2}} + \left( \partial_{t} \A + \nabla \varphi \right) \times (\nabla \times \A)   \right]\, dx
\]
and velocity of the ergocenter $\vec{V}$. Finally the gauge action gives one more conservation law, the hylomorphic charge
\begin{equation} \label{charge-kgm-p}
\Ch = \int_{\R^3}\, \rho \, dx = q\, \int\, \left( \partial_{t}S +q\varphi \right) u^{2}\, dx\, .
\end{equation}

The existence of soliton and vortices solutions to equations (\ref{prima-kgm-p1})-(\ref{terza-kgm-p}) has been proved in \cite{bf-vort,bf-stab-g} using the ansatz \eqref{ansatz-nlkg-vort} 
\[
\psi(t,x) = u(x)\, e^{i( \ell \theta(y)-\omega t)}\, .
\]
These solutions have non-vanishing matter angular momentum (see \cite{bf-vort})
\[
\vec{L}_{m} := \int_{\R^{3}}\, \ell u^{2}\, (-\omega + q\varphi)\, (\vec{x} \times \nabla \theta)\, dx
\]
when $\ell \not= 0$. We recall the notation $x=(y,z)\in \R^{2}\times \R$ and $r=\sqrt{y_{1}^{2}+y_{2}^{2}}$. Benci and Fortunato proved that
\begin{theorem}[\cite{bf-vort,bf-stab-g}] \label{main-bf}
Let $W$ satisfy (W0)-(W3) of Section \ref{sec:nlkg}. Then for all $\ell \in \Z$ there exists $q_0 >0$ such that for every $q \in (0,q_0)$ the system (\ref{prima-kgm-p1})-(\ref{terza-kgm-p}) admits a finite energy solution $(u,\omega,\varphi,\A)$ in the sense of distributions with: $u=u(r,z)\not\equiv 0$; $\omega >0$; $\varphi=\varphi(r,z)\not\equiv 0$; $\A = a(r,z)\, \nabla \vartheta$ with $\A \equiv 0$ if and only if $\ell = 0$. Moreover, if $\ell=0$ these solutions are orbitally stable. 
\end{theorem}
This theorem shows the existence of solitons and vortices for small enough interaction between the matter and the gauge field as quantified by the parameter $q$. In the next subsection we give more details of the proof of Theorem \ref{main-bf} for solitary waves, that is for $\ell=0$, studying the dependence on the charge \eqref{charge-kgm-p}, showing the existence of solitary waves for arbitrarily large $q$ or electric charge.

\subsection{Solitary waves in Abelian gauge theories} \label{sec:nr-abelian}

Here we follow the approach in \cite{bf} (see also \cite{mugnai}). We look for solitary waves solutions to the system (\ref{prima-kgm-p1})-(\ref{terza-kgm-p}) using the ansatz \eqref{ansatz-nlkg}
\[
\psi(t,x)= u(x) e^{-i\omega t}
\]
with $\omega\not= 0$, so that by Theorem \ref{main-bf} we also have $\A\equiv 0$, and we also assume $\varphi= \varphi(x)$ and introduce the notation $\phi(x) = \frac{\varphi(x)}{\omega}$. Hence the Klein-Gordon-Maxwell system reduces to the equations
\begin{align}
& -\triangle u - \omega^{2} (q\phi - 1)^{2} u + W'(u) = 0 \label{uno-sw-ab} \\ 
& -\triangle \phi +q (q\phi-1)u^{2}=0 \label{due-sw-ab}
\end{align}
with equations \eqref{prima-kgm-p2} and \eqref{terza-kgm-p} being identically satisfied. So we consider the space of solitary waves
\[
X_{S}:= \set{ (u,\omega,\phi) \in H^{1}(\R^{3},\R^{+})\times \R \times \dot H^{1}(\R^{3},\R)}
\]
which is embedded into $X$ by
\[
X_{S}\ni (u,\omega,\phi)\, \mapsto\, \Big( u(x) e^{-i\omega t},\, \omega \phi(x), \, \mathbf{0} \Big) \in X
\]
Energy and charge on $X_{S}$ are given by
\[
\tilde E(u,\omega,\phi) := \E|_{X_{S}} = \int_{\R^{3}}\, \Big( \frac 12 |\nabla u|^{2} + \frac 12 \omega^{2} (q\phi -1)^{2} u^{2} + W(u) +\frac 12 \omega^{2} |\nabla \phi|^{2} \Big)\, dx  
\]
\[
\tilde C(u,\omega,\phi) := \Ch|_{X_{S}} = q\, \int_{\R^{3}}\, \omega (q\phi -1) u^{2}\, dx  
\]
In their approach to system \eqref{uno-sw-ab}-\eqref{due-sw-ab}, Benci and Fortunato first prove that for any $u\in H^{1}$ there exists a unique solution $\phi_{u} \in \dot H^{1}$ to \eqref{due-sw-ab}, with the map
\[
H^{1}\ni u \mapsto \phi_{u} \in \dot H^{1}
\]
being of class $C^{1}$, and
\begin{equation}\label{alt}
0\le \phi_{u}(x) \le \frac 1q
\end{equation}
Hence they introduce on $H^{1}$ the $C^{1}$ functional
\begin{equation}\label{kappa}
K(u) := \int_{\R^{3}}\, \Big( |\nabla \phi_{u}|^{2} + (q\phi_{u} -1)^{2} u^{2}\Big)\, dx = \int_{\R^{3}}\, (1-q\phi_{u})\, u^{2}\, dx
\end{equation}
which satisfies (cfr. Lemma 8 in \cite{bf} and Lemma 2.1 in \cite{mugnai})
\[
K'(u) = 2u\, (1-q\phi_{u})^{2}
\]
Hence if we consider the reduced energy and charge
\[
E(u,\omega) := \tilde E(u,\omega,\phi_{u}) = \int_{\R^{3}}\, \Big( \frac 12 |\nabla u|^{2} + W(u) \Big)\, dx + \frac 12 \omega^{2} K(u)
\]
\[
C(u,\omega) := \tilde C(u,\omega,\phi_{u}) = - q\omega K(u)
\]
we get
\begin{proposition}[\cite{bf}] \label{critvinc-ab}
A triple $(u,\omega,\phi)\in X_{S}$ is a solution of the system \eqref{uno-sw-ab}-\eqref{due-sw-ab} if and only if $\phi=\phi_{u}$ solves equation \eqref{due-sw-ab} and the couple $(u,\omega)$ is a critical point of the energy $E(u,\omega)$ constrained to the manifold
\[
\Sigma^{S}_{\sigma}:= \set{(u,\omega) \in H^{1}\times \R\, : \, C(u,\omega) = q\sigma}\, .
\]
\end{proposition}
Here $q$ is a fixed parameter and, without loss of generality, we assume $\sigma>0$ and $\omega<0$, since $K(u)\ge 0$ by \eqref{alt}. Using the notation
\[
W(s) = \frac 12 m^{2} s^{2} + R(s)
\]
with $R(s)$ satisfying (W0)-(W3) of Section \ref{sec:nlkg}, we define
\begin{equation}\label{j-ab}
J(u) := \int_{\R^{3}}\, \Big( \frac 12 |\nabla u|^{2} + R(u) + \frac 12 m^{2} q\phi_{u}\, u^{2} \Big)\, dx 
\end{equation}
and we write the energy $E(u,\omega)$ on $\Sigma^{S}_{\sigma}$ as
\[
E_{\sigma}(u):= E|_{\Sigma^{S}_{\sigma}} = J(u) + \frac 12 \Big( m^{2} K(u) + \frac{\sigma^{2}}{K(u)} \Big)
\]
By Proposition \ref{critvinc-ab} we are reduced as in Section \ref{sec:nlkg} to look for critical points of $E(u,\omega)$ constrained to $\Sigma^{S}_{\sigma}$. In \cite{bf} and \cite{bf-stab-g}, Benci and Fortunato show that the analogous of Theorem \ref{orbstab-min} holds. Hence the existence of a soliton solution to system \eqref{uno-sw-ab}-\eqref{due-sw-ab} is implied by the existence of a point of local minimum for $E(u,\omega)$ constrained to $\Sigma^{S}_{\sigma}$. Benci and Fortunato use the so-called \emph{hylomorphy ratio} $\Lambda(u)$ given by
\begin{equation}\label{hylom-ab}
\Lambda(u,\omega) := \frac{E(u,\omega)}{-\omega K(u)} = \frac{E_{\sigma}(u)}{\sigma}\, ,
\end{equation}
introduced in \cite{hylo-sol}, and show that if there exists $(\bar u,\bar \omega) \in \Sigma^{S}_{\sigma}$ such that $\Lambda(\bar u,\bar \omega) < m$, then, assuming (W0)-(W3), $E(u,\omega)$ admits a global minimizer on the manifold $\Sigma^{S}_{\sigma}$, hence there exists a soliton solution to system \eqref{uno-sw-ab}-\eqref{due-sw-ab} with electric charge $\Ch=q\sigma$. By \eqref{hylom-ab} this is equivalent to show that, using $\sigma$ as a parameter, there exists $\bar u\in H^{1}$ such that $E_{\sigma}(\bar u)< m\sigma$ (cfr. Lemma 19 in \cite{bf-stab-g}).

We now argue as in \cite{bon10} to give more information on the values of $\sigma$ and $q$ for which we have a soliton solution to system \eqref{uno-sw-ab}-\eqref{due-sw-ab} with electric charge $\Ch = q\sigma$. First we prove the analogous of Theorem \ref{beh-carica}-(i). Let
\[
J^{-} := \set{u\in H^{1}(\R^{3},\R^{+})\, :\, J(u)<0}
\]
with $J(u)$ defined in \eqref{j-ab}. Then
\begin{proposition}\label{prop-bon-sigma}
Under assumptions (W0)-(W3), $E(u,\omega)$ admits a point of global minimum on $\Sigma^{S}_{\sigma}$ for all $\sigma\in (\sigma_{g}, \sigma_{_{G}})$, where
\begin{align*}
& \sigma_{g}:= \inf_{u\in J^{-}} \Big( mK(u) - \sqrt{2K(u) |J(u)|}\Big) \\ 
& \sigma_{_{G}}:= \sup_{u\in J^{-}} \Big( mK(u) + \sqrt{2K(u) |J(u)|}\Big)
\end{align*}
and $\sigma_{g}=\sigma_{_{G}}=+\infty$ if $J^{-}=\emptyset$.
\end{proposition}
\begin{proof}
We follow the proof of Proposition 2.4 in \cite{bon10}. By the results by Benci and Fortunato, we need only to show that for all $\sigma \in (\sigma_{g}, \sigma_{_{G}})$ we have
\begin{equation}\label{scopo}
\inf_{u\in H^{1}}\, \frac{E_{\sigma}(u)}{\sigma} = \inf_{u\in H^{1}}\, \Big[ \frac 1\sigma J(u) +  \frac 12 \Big( \frac{m^{2}}{\sigma} K(u) + \frac{\sigma}{K(u)} \Big) \Big] < m
\end{equation}
We recall that $K(u)\ge 0$ by \eqref{alt}, and $K(u)\not= 0$ for all $u\in H^{1}$ because $\varphi_{u}\in \dot H^{1}$, hence
\[
\inf_{u\in H^{1}} \frac 12 \Big( \frac{m^{2}}{\sigma} K(u) + \frac{\sigma}{K(u)} \Big) \ge m
\]
It follows that $E_{\sigma}<m\sigma$ implies $J(u)<0$, hence we have to consider only functions in $J^{-}$. Moreover, writing $\Lambda$ in \eqref{hylom-ab} as a function of $u$ and $\sigma$, from basic algebra it follows that
\[
\Lambda(u,\sigma)= \frac{E_{\sigma}(u)}{\sigma} \ge m \quad \Leftrightarrow \quad \sigma \in \R^{+}\setminus \Big(\sigma_{g}(u), \sigma_{_{G}}(u)\Big)
\]
with
\begin{align}
& \sigma_{g}(u):= mK(u) - \sqrt{2K(u) |J(u)|} \label{s-picc} \\ 
& \sigma_{_{G}}(u):= mK(u) + \sqrt{2K(u) |J(u)|} \label{s-gran}
\end{align}
Whence $E_{\sigma}\ge m\sigma$ for all $u\in J^{-}$ if and only if $\sigma \in \R^{+}\setminus \Big(\sigma_{g}, \sigma_{_{G}}\Big)$, where we have used continuity of the functions $\sigma_{g, _{G}}(u)$ and the non-vanishing of $K$ to show that
\[
\bigcup_{u\in H^{1}}\, \Big(\sigma_{g}(u), \sigma_{_{G}}(u)\Big) = \Big(\sigma_{g}, \sigma_{_{G}}\Big)
\]
Inequality \eqref{scopo} for $\sigma\in (\sigma_{g}, \sigma_{_{G}})$ is proved.
\end{proof}
Proposition \ref{prop-bon-sigma} implies that soliton solutions exist for all electric charges $\Ch \in (q\sigma_{g}, q\sigma_{_{G}})$ if $J^{-}\not= \emptyset$, where we remark that $\sigma_{g,_{G}}$ depend on $q$ since $K(u)$ and $J(u)$ do. Benci and Fortunato have shown that if $q$ is small enough then $J^{-} \not= \emptyset$, giving no information on the values of $\sigma_{g,_{G}}$. They conjecture that $q\sigma_{_{G}}<\infty$.

We now study the possible values of $q$ for which $J^{-} \not= \emptyset$. Let us denote by $c_{3}$ the best constant in the Gagliardo-Nirenberg inequality in $\R^{3}$, that is
\begin{equation}\label{uso}
c_{3} \Big( \int_{\R^{3}}\, |\phi|^{6}\, dx \Big)^{\frac 13} \le \int_{\R^{3}}\, |\nabla \phi|^{2}\, dx
\end{equation}
for all $\phi \in \dot H^{1}$.
\begin{proposition}\label{nuove-q}
For any fixed $q>0$ we assume that $W$ satisfies (W0)-(W3) and that there exist $s_{1},r>0$ and $h\in (0,1)$ such that
\begin{itemize}
\item[(W4)] $W$ is non-decreasing in $(0,s_{1})$ and 
\begin{equation}\label{ale}
-\frac 12 m^{2} s_{1}^{2} \le R(s_{1}) < \frac 12 s_{1}^{2} \Big[ (1+m^{2}h^{2}) \frac{r^{3}}{(r+1)^{3}}\, - (1+m^{2}) \Big]
\end{equation}
\end{itemize}
and
\begin{equation}\label{gagl-nir}
\Big( \frac{c_{3}}{48^{\frac 13} \pi^{\frac 23}} \Big)^{\frac 12}\, \frac{1-h}{qh} > s_{1}r\, .
\end{equation}
Then there exists $u\in H^{1}$ such that $J(u) <0$.
\end{proposition}
\begin{proof}
We first analyze the term $\int m^{2}q\phi_{u}u^{2}$ in $J(u)$. We recall from \cite{bf} that, for any fixed $u\in H^{1}$, the solution $\phi_{u}$ of \eqref{due-sw-ab} is the unique critical point of the functional
\[
K(u,\phi) = \int_{\R^{3}}\, \Big( |\nabla \phi|^{2} + (q\phi -1)^{2} u^{2}\Big)\, dx\, ,
\]
in particular $\phi_{u}$ is the global minimizer of $K(u,\phi)$. The functional $K(u)$ defined in \eqref{kappa} satisfies $K(u) = K(u,\phi_{u})$, and it follows that
\[
\frac 12 m^{2} \int_{R^{3}} q\phi_{u}u^{2}\, dx = \frac 12 m^{2} \Big(\|u\|_{2}^{2} - K(u)\Big)
\]
Hence, letting
\begin{equation}\label{i}
I(u):= \min_{\phi\in \dot H^{1}} \Big(K(u,\phi) - \|u\|^{2}_{2}\Big) = K(u) - \|u\|^{2}_{2}
\end{equation}
we write
\[
\frac 12 m^{2} \int_{R^{3}} q\phi_{u}u^{2}\, dx = -\frac 12 m^{2} I(u) = -\frac 12 m^{2} \min_{\phi\in \dot H^{1}} \Big(K(u,\phi) - \|u\|^{2}_{2}\Big)\, .
\]
Let $s_{1},r>0$ and $h\in (0,1)$ such that (W4) and \eqref{gagl-nir} are satisfied. Then we define
\begin{equation}\label{ur}
u_{r}(x):= \left\{ \begin{array}{ll} s_{1}\, , & \text{if }\, |x|\le r\\ s_{1}(r+1-|x|)\, , & \text{if }\, r\le |x|\le r+1\\ 0\, , & \text{if }\, |x|\ge r+1\end{array} \right.
\end{equation}
for which
\begin{equation}\label{2ur}
\|u_{r}\|_{2}^{2} = \frac 43 \pi s_{1}^{2} r^{3} + 4\pi s_{1}^{2} \int_{r}^{r+1} (r+1-t)^{2}\,t^{2} dt
\end{equation}
Then we claim that
\begin{equation}\label{claim}
I(u_{r}) \ge (h^{2}-1) \|u_{r}\|_{2}^{2}\, .
\end{equation}
We first show that \eqref{claim} implies $J(u_{r})<0$. We have
\begin{align*}
J(u_{r}) = & \int_{\R^{3}}\, \Big( \frac 12 |\nabla u_{r}|^{2} + R(u_{r})\Big)\, dx -\frac 12 m^{2}I(u_{r}) \le\\ 
\le & \int_{\R^{3}}\, \Big( \frac 12 |\nabla u_{r}|^{2} + R(u_{r})\Big)\, dx -\frac 12 m^{2} (h^{2}-1) \|u_{r}\|_{2}^{2} = \\
= & \int_{\R^{3}}\, \Big( \frac 12 |\nabla u_{r}|^{2} + W(u_{r})\Big)\, dx -\frac 12 m^{2} h^{2} \|u_{r}\|_{2}^{2}
\end{align*}
Now, using \eqref{ur}, the fact that $W$ is non-decreasing in $(0,s_{1})$ by (W4) and \eqref{2ur}, we have
\[
J(u_{r}) \le \frac{2\pi}{3} \, s_{1}^{2}\Big( (r+1)^{3}-r^{3} \Big) + \frac{4\pi}{3} \, W(s_{1}) (r+1)^{3} - \frac{2\pi}{3} \,m^{2}h^{2}s_{1}^{2}r^{3}
\]
which implies $J(u_{r})<0$ by \eqref{ale}.

It remains to prove \eqref{claim}. First, since $u_{r}$ is radially symmetric, by Lemma 6 in \cite{bf} and Proposition 2.2 in \cite{mugnai}, the minimum $K(u)$ will be achieved for $\phi_{u}$ radially symmetric. Hence by \eqref{i}, and the Gagliardo-Nirenberg inequality \eqref{uso}, we have
\[
I(u_{r}) \ge \inf_{\phi\in \dot H^{1}_{r}} \Big[ c_{3} \Big( \int_{\R^{3}}\, |\phi|^{6}\, dx \Big)^{\frac 13} + \int_{\R^{3}} (q\phi-1)^{2}\, u_{r}^{2}\, dx\Big] - \|u_{r}\|_{2}^{2}
\]
where $\dot H^{1}_{r}$ is the set of radially symmetric functions in $\dot H^{1}$. For any $\phi\in \dot H^{1}_{r}$ we define
\[
\rho_{\phi} := \inf\set{ \rho>0\, :\, \phi(x) \le \frac 1q (1-h) \, \text{for all }\, |x|>\rho}
\]
and write
\[
\int_{\R^{3}}\, |\phi|^{6}\, dx \ge \int_{B(0,\rho_{\phi})}\, \Big[ \frac 1q (1-h) \Big]^{6}\, dx =
\frac{4\pi}{3} \frac{(1-h)^{6}}{q^{6}} \rho_{\phi}^{3}
\]
where $B(0,\rho)$ is the ball in $\R^{3}$ centered in 0 and of radius $\rho$, and, using \eqref{ur},
\begin{align}
& \int_{\R^{3}} (q\phi-1)^{2}\, u_{r}^{2}\, dx \ge \int_{\R^{3}\setminus B(0,\rho_{\phi})}\, h^{2}\, u_{r}^{2} \, dx =\\
& = \left\{ \begin{array}{ll} \frac{4\pi}{3} h^{2} s_{1}^{2} \Big( r^{3} - \rho_{\phi}^{3} \Big) + 4\pi h^{2} s_{1}^{2} \int_{r}^{r+1} (r+1-t)^{2}\,t^{2} dt\, , & \text{if }\, \rho_{\phi}<r \\ 4\pi h^{2} s_{1}^{2} \int_{\rho_{\phi}}^{r+1} (r+1-t)^{2}\,t^{2} dt\, , & \text{if }\, r\le \rho_{\phi}<r+1\\ 0\, , & \text{if }\, \rho_{\phi}\ge r+1
\end{array} \right. \label{fr}
\end{align}
Hence, if we define the function $f(\rho)$ on $\R^{+}$ by
\[
f(\rho) := c_{3} \Big(\frac{4\pi}{3}\Big)^{\frac 13} \frac{(1-h)^{2}}{q^{2}}\, \rho + \int_{\R^{3}\setminus B(0,\rho)}\, h^{2}\, u_{r}^{2} \, dx 
\]
it follows that
\[
I(u_{r}) \ge \inf_{\R^{+}}\, f(\rho) - \|u_{r}\|_{2}^{2}
\]
The function $f$ is continuous and by \eqref{fr}
\[
f'(\rho) = \left\{ \begin{array}{ll} c_{3} \Big(\frac{4\pi}{3}\Big)^{\frac 13} \frac{(1-h)^{2}}{q^{2}} - 4\pi h^{2} s_{1}^{2} \rho^{2} \, , & \text{if }\, \rho<r \\ c_{3} \Big(\frac{4\pi}{3}\Big)^{\frac 13} \frac{(1-h)^{2}}{q^{2}} - 4\pi h^{2} s_{1}^{2} (r+1-\rho)^{2}\,\rho^{2}\, , & \text{if }\, r\le \rho<r+1\\ c_{3} \Big(\frac{4\pi}{3}\Big)^{\frac 13} \frac{(1-h)^{2}}{q^{2}}\, , & \text{if }\, \rho\ge r+1
\end{array} \right.
\]
Then if \eqref{gagl-nir} holds there are no critical points and $f$ is increasing in $(0,r)$. Moreover,
\[
4\pi h^{2} s_{1}^{2} (r+1-\rho)^{2}\,\rho^{2} \le 4\pi h^{2} s_{1}^{2} r^{2} \quad  \text{if }\, r\le \rho<r+1
\]
hence \eqref{gagl-nir} implies that $f'(\rho)>0$ also in $(r,r+1)$. It follows that $f$ is an increasing function. Hence
\[
I(u_{r}) \ge \inf_{\R^{+}}\, f(\rho) - \|u_{r}\|_{2}^{2} = f(0) - \|u_{r}\|_{2}^{2} = (h^{2}-1) \|u_{r}\|_{2}^{2}
\]
and \eqref{claim} is proved. This finishes the proof of the proposition.
\end{proof}
We now discuss assumptions \eqref{ale} and \eqref{gagl-nir}. First of all \eqref{ale} can be written as
\begin{equation}\label{ale1}
0\le W(s_{1}) < \frac 12 s_{1}^{2} \Big[ (1+m^{2}h^{2}) \frac{r^{3}}{(r+1)^{3}}\, - 1 \Big]
\end{equation}
where the inequality on the left is satisfied by (W1), hence it is necessary that
\begin{equation}\label{nec-ale}
m^{2}h^{2}r^{3}-3r^{2}-3r-1>0\, .
\end{equation}
So, for example, if we fix $h$ and $r$ such that \eqref{nec-ale} is satisfied, then we choose $s_{1}$ so that \eqref{gagl-nir} is satisfied, and impose \eqref{ale1} on $W$ at that $s_{1}$. Notice that \eqref{ale1} is not in contradiction with (W3) which prescribes the behaviour of $R$ at $s=0$.

Putting together Propositions \ref{critvinc-ab}, \ref{prop-bon-sigma} and \ref{nuove-q}, we prove that
\begin{corollary}\label{sol-ogni-q}
For any fixed $q>0$, let $W(s), s_{1}, r$ and $h$ satisfy (W0)-(W4), \eqref{ale} and \eqref{gagl-nir}. Then there exist solitons solutions to system \eqref{uno-sw-ab}-\eqref{due-sw-ab} for any electric charge $\Ch\in (q \sigma_{g}(u_{r}),\, q\sigma_{_{G}}(u_{r}))$, where $u_{r}(x)$ is defined in \eqref{ur} and $\sigma_{g,_{G}}(u_{r})$ are given by \eqref{s-picc} and \eqref{s-gran}.
\end{corollary}
\begin{proof}
By Proposition \ref{nuove-q} it holds $u_{r}\in J^{-}$, whence $\sigma_{g,_{G}}$ defined in Proposition \ref{prop-bon-sigma} satisfy
\[
\sigma_{g} \le \sigma_{g}(u_{r})\, ,\qquad \sigma_{_{G}}\ge \sigma_{_{G}}(u_{r})\, .
\]
Hence $E(u,\omega)$ admits a point of global minimum on $\Sigma_{\sigma}^{S}$ for all \[
\sigma \in (\sigma_{g}(u_{r}),\, \sigma_{_{G}}(u_{r}))\, ,
\]
and there exists a triple $(u,\omega,\phi_{u})$ which is a solution to the system \eqref{uno-sw-ab}-\eqref{due-sw-ab} for all $\Ch \in (q \sigma_{g}(u_{r}),\, q\sigma_{_{G}}(u_{r}))$. That this solution is a soliton is given by Benci and Fortunato results in \cite{bf-stab-g}.
\end{proof}
Finally we show that it is possible to have soliton solutions with arbitrarily large electric charge by changing the interaction parameter $q$.
\begin{theorem}\label{mio-ab}
Let $W$ be a non-decreasing function satisfying (W0)-(W3) of Section \ref{sec:nlkg}. Then for any $\bar \Ch>0$ there exists $q>0$ such that the system \eqref{uno-sw-ab}-\eqref{due-sw-ab} admits a soliton solution with electric charge $\Ch\ge \bar \Ch$.
\end{theorem}
\begin{proof}
Let $s_{0}$ given in (W2) and fix $s_{1}=s_{0}$. Then
\[
\lambda := \frac{W(s_{0})}{\frac 12 s_{0}^{2}} < m^{2}
\]
and there exists $\alpha>0$ such that
\begin{equation}\label{alpha}
\lambda < m^{2}(1-\alpha) - \alpha\, .
\end{equation}
We now choose $r>0$ such that
\[
\frac{r^{3}}{(r+1)^{3}} > 1-\alpha \quad \Leftrightarrow \quad r > \frac{1}{(1-\alpha)^{-\frac 13} -1}\, ,
\]
and $h$ such that \eqref{ale} is satisfied, that is, using \eqref{ale1},
\begin{equation}\label{passo1}
\lambda < (1+m^{2}h^{2}) (1-\alpha) -1 < (1+m^{2}h^{2}) \frac{r^{3}}{(r+1)^{3}} -1
\end{equation}
It is possible to choose such $h$, since \eqref{passo1} implies
\[
h^{2}\in \Big( \frac{\lambda+\alpha}{m^{2}(1-\alpha)}, \, 1 \Big)
\]
which is consistent by \eqref{alpha}.

So far we have fixed $s_{1}$ and $h$, and have found that (W4) of Proposition \ref{nuove-q} is satisfied for $r$ big enough. To satisfy also \eqref{gagl-nir} we can still move $q$. So for any $r$ let us choose
\begin{equation}\label{q}
q = \frac 12 \Big( \frac{c_{3}}{48^{\frac 13} \pi^{\frac 23}} \Big)^{\frac 12} \frac{1-h}{hs_{1}r} < \Big( \frac{c_{3}}{48^{\frac 13} \pi^{\frac 23}} \Big)^{\frac 12} \frac{1-h}{hs_{1}r}
\end{equation}
so that \eqref{gagl-nir} is satisfied, and we still can move $r$. Then we can apply Corollary \ref{sol-ogni-q} and find a soliton solution with electric charge
\[
\Ch = q m K(u_{r}) \in \Big(q \sigma_{g}(u_{r}),\, q\sigma_{_{G}}(u_{r})\Big)
\]
where $u_{r}$ is defined in \eqref{ur}. To finish the proof of the theorem, we use consecutively  \eqref{i}, \eqref{claim} and \eqref{2ur} to show that
\[
qmK(u_{r}) = qm \Big(I(u_{r}) + \| u_{r}\|_{2}^{2} \Big) \ge qmh^{2}\, \| u_{r}\|_{2}^{2} \ge \frac{4\pi}{3} qmh^{2} s_{1}^{2}r^{3}
\]
Finally from \eqref{q}, we get
\[
qmK(u_{r}) \ge \frac{2\pi}{3} \Big( \frac{c_{3}}{48^{\frac 13} \pi^{\frac 23}} \Big)^{\frac 12} mh(1-h)s_{1}r^{2}
\]
where $h$ and $s_{1}$ are fixed. Hence for any $\bar \Ch$, we can choose $r$ big enough so that
\[
r > \frac{1}{(1-\alpha)^{-\frac 13} -1} \quad \text{and} \quad qmK(u_{r}) \ge \bar \Ch
\]
and the proof is finished.
\end{proof}

\section{Local gauge theory: the non-Abelian case} \label{sec:nonabelian}

In this section we briefly review the results proved in \cite{bb}. We consider the case $N=2$ with non-Abelian gauge group $G=SU(2)$, so that $\psi(t,x) \in \C^{2}$ and $\Gamma=(\Gamma_{j})$ with $\Gamma_{j}(t,x) \in \lie = \su{2}$. The real Lie algebra $\su{2}$ is generated by $i$ times the Pauli matrices
\[
\tau_{1}:= i \sigma_{x} = \matrdpd{0}{i}{\imath}{0} \quad
\tau_{2}:= i \sigma_{y} = \matrdpd{0}{1}{-1}{0} \quad
\tau_{3}:= i \sigma_{z} = \matrdpd{i}{0}{0}{-i}
\]
By the properties of compact Lie groups, the exponential map $\exp : \su{2} \to SU(2)$
is surjective and for each $g\in SU(2)$ there exists a triple $S=(S_1,S_2,S_3) \in \R^3$ with $\sum_{j=1}^3 S_j^2 \le \pi^2$ such that
\[
g = \exp(S_1 \tau_1 + S_2 \tau_2 + S_3 \tau_3)
\]
and it is unique when $\sum_{j=1}^3 S_j^2 < \pi^2$. Given $(S_1,S_2,S_3) \in \R^3$ we introduce the notation
\begin{equation} \label{fasi-su2}
\SU(t,x) := S_1(t,x)\, \tau_1 + S_2(t,x)\, \tau_2 + S_3(t,x)\, \tau_3, \quad |\SU|^2:= |S|^2 = \sum_{j=1}^3 S_j^2
\end{equation}
and the operations
\begin{align}
& \partial_j \SU := \partial_j S_1(t,x)\, \tau_1 + \partial_j S_2(t,x)\, \tau_2 + \partial_j S_3(t,x)\, \tau_3 \label{der-pf} \\
& \SU \times \tilde \SU := (S \times \tilde S)_1\, \tau_1 + (S \times \tilde S)_2\, \tau_2+ (S \times \tilde S)_3\, \tau_3 = -\frac 12\, [\SU,\tilde \SU] \label{lb-pf} \\
& \SU \cdot \tilde \SU := S_1 \tilde S_1 + S_2 \tilde S_2 + S_3 \tilde S_3 = \frac 12 \langle \SU, \tilde \SU \rangle \label{tr-pf} \\
& \SU \, \tilde \SU = - \SU \cdot \tilde \SU - \SU \times \tilde \SU \label{pr-pf}
\end{align}
where $[\cdot,\cdot]$ is the standard Lie bracket and in the last equation on the left hand side we use the usual matrix product. Finally for the gauge potentials with abuse of notation we write
\begin{equation} \label{gf-pf}
\Gamma_{j} := \gamma_{j,1}\, \tau_{1} + \gamma_{j,2}\, \tau_{2} + \gamma_{j,3}\, \tau_{3}, \qquad j=0,1,2,3
\end{equation}
as in (\ref{fasi-su2}), and extend to $\Gamma_{j}$ the operations (\ref{der-pf}) and (\ref{lb-pf}). We then introduce the polar form for matter fields
\[
\psi(t,x) = u(t,x) \, e^{\SU(t,x)}\, \psi_0, \qquad u\in \R^+,\, |\SU(t,x)| \le \pi
\]
for a fixed vector $\psi_0 \in \C^2$, $|\psi_{0}|_{\C^{2}}=1$. We first have
\begin{lemma}[\cite{bb}] \label{lem:derivata-esp}
For all $\SU \in \su{2}$ with regular functions $S_i(t,x)$, it holds
\begin{equation} \label{form-esp}
\partial_{j} \exp(\SU) = C(\SU, \partial_{j}\SU)\, \exp(\SU)
\end{equation}
with
$$
C(\SU, \partial_{j}\SU) := \partial_j \SU + \frac 12 \, (1-\cos 2)\, (\partial_j \SU \times \SU) + \frac 12\, (2-\sin 2)\, ((\partial_j \SU \times \SU) \times \SU) \in \su{2}
$$
\end{lemma}
Using (\ref{form-esp}) the covariant derivatives \eqref{deriv-cov} write
\[
D_{j} \left( u\, e^{\SU}\, \psi_{0} \right) = \left[ \partial_{j} u + u\, C(\SU,\partial_{j}\SU) + q u\, \Gamma_{j} \right]\, e^{\SU}\, \psi_{0}
\]
and
\[
| D_{j} \left( u\, e^{\SU}\, \psi_{0} \right) |_{\C^{2}}^{2} = |\partial_{j} u|^{2} + u^{2}\, \left| C(\SU,\partial_{j}\SU) + q\, \Gamma_{j} \right|^{2}
\]
Hence we have for $\LL_{0}$ defined in \eqref{lag-matter}
\[
\begin{aligned}
\LL_{0} = & \frac 12 |\partial_{t} u|^{2} - \frac 12 |\nabla u|^{2} - W(u) + \\
 & + \frac 12 \, u^{2}\, \left[ \left| C(\SU,\partial_{t}\SU) - q\, \Gamma_{0} \right|^{2} - \sum_{j=1}^{3} \left| C(\SU,\partial_{j}\SU) + q\, \Gamma_{j} \right|^{2} \right]
\end{aligned}
\]
Moreover, since $F_{kj} \in \su{2}$ we have
\[
\| F_{kj} \|^{2} = - trace(F_{kj}^{2}) = 2 \left| \partial_{k} \Gamma_{j} - \partial_{j} \Gamma_{k} -2q\, (\Gamma_{k} \times \Gamma_{j}) \right|^{2}
\]
where for $\Gamma_{j}$ we have used notation (\ref{gf-pf}) and (\ref{lb-pf}). Hence from \eqref{lag-gauge} we get
\[
\begin{aligned}
\LL_{1} = & \sum_{j=1}^{3}\, \left| \partial_{t} \Gamma_{j} + \partial_{j} \Gamma_{0} + 2q\, (\Gamma_{0} \times \Gamma_{j})\right|^{2} - \\  
& - \frac 12\, \sum_{k,j=1}^{3}\, \left| \partial_{k} \Gamma_{j} - \partial_{j} \Gamma_{k} -2q\, (\Gamma_{k} \times \Gamma_{j}) \right|^{2}
\end{aligned}
\]
Hence we get the \emph{Yang-Mills-Higgs} system of equations\footnote{This system does not coincide with classical Yang-Mills-Higgs equations because of the properties of the nonlinear term $W$. For a discussion of this remark we refer to \cite{bb}.}, which is the analogous of system \eqref{prima-kgm-p1}-\eqref{terza-kgm-p}, and is given by two equations describing the evolution of the matter field
\begin{align} 
& \partial_{t}^{2} u - \triangle u + \Big[ \sum_{j=1}^{3} | C(\SU, \partial_{j} \SU)+q\Gamma_{j} |^{2} - | C(\SU, \partial_{t} \SU) -q\Gamma_{0} |^{2}\Big]\, u   + W'(u) = 0 \label{epsi-pf}  \\
& D_{0} \Big( \left( C(\SU, \partial_{0} \SU) +q\Gamma_{0} \right) \, u^{2} \Big) - \sum_{j=1}^{3}\, D_{j} \Big( \left( C(\SU, \partial_{j} \SU) +q\Gamma_{j} \right) \, u^{2} \Big) = 0 \label{epsi-pf-cont}
\end{align}
and a system of four equations for the gauge field
\begin{align}
& 2 \, \sum_{j=1}^{3}\, D_{j} F_{0j} - q\, u^{2} [C(\SU, \partial_{t} \SU) - q\Gamma_{0}] =0 \label{egamma0-pf}\\
& 2\, D_{0} F_{0j} - 2\, \sum_{\ell\not= j}\, D_{\ell} F_{\ell j} + q\, u^{2} [C(\SU, \partial_{j} \SU) + q\Gamma_{j}] = 0, \qquad j=1,2,3 \label{egammaj-pf}
\end{align}
For fields which vanish at infinity sufficiently fast, energy and charge have the form
\[
\E = \int_{\R^3}\, \left[ 
\begin{array}{l} \frac 12 |\partial_{t} u|^{2} + \frac 12 |\nabla u|^{2} +  W(u) + \\ + \frac 12 \, u^{2}\, \left[ \left| C(\SU,\partial_{t}\SU) - q\, \Gamma_{0} \right|^{2} + \sum_{j=1}^{3} \left| C(\SU,\partial_{j}\SU) + q\, \Gamma_{j} \right|^{2} \right] + \\[0.2cm] + \sum_{j=1}^{3}\, \| F_{0j}\|^{2} + \frac 12\, \sum_{k,j=1}^{3}\, \| F_{kj}\|^{2} \end{array} \right]\, dx
\]
\[
\Ch = \int_{\R^3}\, \Big[ u^2\, \left( C(\SU, \partial_t \SU) -q \Gamma_0\right)) - 2\, \sum_{j=1}^3\, [\Gamma_j, F_{0j}] \Big]\ dx \in \su{2}
\]
We now introduce the ansatz analogous to \eqref{ansatz-nlkg-vort} to find solitary waves solutions for system (\ref{epsi-pf})-(\ref{egammaj-pf}), that is
\[
\psi(t,x) = u(r, z)\, e^{S(t,x)\, \tau_{m}}\, \psi_{0}, \qquad u\in \R^{+},\ m=1,2,3,\ |S(t,x)|\le \pi
\]
where $S(t,x) = \ell \vartheta(y) - \omega \, t$, with $\omega \in \R$, $\ell \in \Z$ and $\vartheta(y)$ defined in \eqref{angolo}. For the gauge field we assume analogously that 
\[
\Gamma_{0} = \gamma_{0}( r, z)\ \tau_{m}, \qquad \left( \begin{array}{c} \Gamma_{1} \\ \Gamma_{2} \\ \Gamma_{3} \end{array}\right) = \gamma( r, z)\, \nabla \vartheta \ \tau_{m}
\]
For these functions the matter angular momentum is given by
\[
\vec{L}_{m} = - \int_{\R^{3}}\, \ell u^{2} (\omega + q \gamma_{0}) \, (\vec{x} \times \nabla \vartheta) \, dx
\]
hence it does not vanish if $\ell\not=0$. We find the following equations for the variables $(u, \ell, \omega, \gamma_{0}, \gamma)$, with equation \eqref{epsi-pf-cont} identically satisfied,
\begin{align}
& -\triangle u(x) + \Big[ | (\ell + q\gamma(x)) \nabla \vartheta |^{2} - (\omega+q\gamma_{0}(x))^{2} \Big] \, u + f'(u) =0 \label{matter-su2} \\
& -2\, \triangle \gamma_{0}(x) + q\, (\omega + q\gamma_{0}(x))\, u^{2} =0 \label{gauss-su2} \\
& 2\, \nabla \times \left( \nabla \times \gamma(x) \nabla \vartheta \right) + q\, (\ell + q\gamma(x))\, u^{2}\, \nabla \vartheta =0 \label{rotore-su2}
\end{align}
Our main existence result is
\begin{theorem}[\cite{bb}] \label{main-su2}
Let $W$ satisfy (W0)-(W3) of Section \ref{sec:nlkg}. Then for all $\ell \in \Z$ there exists $q_0 >0$ such that for every $q \in (0,q_0)$ the system (\ref{matter-su2})-(\ref{rotore-su2}) admits a finite energy solution $(u,\omega,\gamma_{0},\gamma)$ in the sense of distributions with: $u=u(r,z)\not\equiv 0$; $\omega >0$; $\gamma_{0}=\gamma_{0}(r,z)\not\equiv 0$; $\gamma= \gamma(r,z)$. Moreover, $\gamma \equiv 0$ if and only if $\ell = 0$.
\end{theorem}
This theorem show the existence of a particular class of solitary waves and vortices for the Yang-Mills-Higgs system for small interaction parameter $q$. Results about stability of these solitary waves and dependence on the charge, analogous to those in Section \ref{sec:nr-abelian}, are not available at the moment. We also hope in the future to prove existence of more general soliton solutions.


\begin{thebibliography}{9}

\bibitem{badrol}
M. Badiale and S. Rolando,
A note on vortices with prescribed charge,
\emph{Adv. Nonlinear Stud.} \textbf{12} (2012), 703--716

\bibitem{nlkg-stab}
J. Bellazzini, V. Benci, C. Bonanno and A. M. Micheletti,
Solitons for the nonlinear Klein-Gordon equation,
\emph{Adv. Nonlinear Stud.} \textbf{10} (2010), 481--499

\bibitem{hylo-sol}
J. Bellazzini, V. Benci, C. Bonanno and E. Sinibaldi, 
Hylomorphic solitons in the nonlinear Klein-Gordon equation, 
\emph{ Dyn. Partial Differ. Equ.} \textbf{6} (2009), 311--334.

\bibitem{hylo-vort}
J. Bellazzini, V. Benci, C. Bonanno and E. Sinibaldi, 
On the existence of hylomorphic vortices in the nonlinear Klein-Gordon equation, 
\emph{ Dyn. Partial Differ. Equ.} \textbf{10} (2013), 1--24.


\bibitem{bb}
V. Benci and C. Bonanno,
Solitary waves and vortices in non-Abelian gauge theories with matter, 
\emph{Adv. Nonlinear Stud.} \textbf{12} (2012), 717--735.

\bibitem{bf} 
V. Benci and D. Fortunato, 
Solitary waves in Abelian gauge theories, 
\emph{Adv. Nonlinear Stud.} \textbf{8} (2008), 327--352.

\bibitem{bf-vort} 
V. Benci and D. Fortunato, 
Spinning Q-balls for the Klein-Gordon-Maxwell equations, 
\emph{Commun. Math. Phys.} \textbf{295} (2010), 639--668.

\bibitem{bf-stab-g} 
V. Benci and D. Fortunato, 
On the existence of stable charged Q-balls, 
\emph{J. Math. Phys.} \textbf{52} (2011), 093701.

\bibitem{bf12}
V. Benci and D. Fortunato,
\emph{Hylomorphic solitons and charged Q-balls: existence and stability}, 
arXiv:1212.3236 [math-ph]

\bibitem{bl}
H. Berestycki and P. L. Lions,
Nonlinear scalar field equations. I. Existence of a ground state,
\emph{Arch. Ration. Mech. Anal.} \textbf{82} (1982), 313--345.

\bibitem{bon10}
C. Bonanno,
Existence and multiplicity of stable bound states for the nonlinear Klein-Gordon equation,
\emph{Nonlinear Anal.} \textbf{72} (2010), 2031--2046.

\bibitem{cazlions}
T. Cazenave and P. L. Lions,
Orbital stability of standing waves for some nonlinear Schr\"odinger equations,
\emph{Commun. Math. Phys.} \textbf{85} (1982), 313--345.

\bibitem{col}
S. Coleman,
Q-balls,
\emph{Nucl. Phys. B} \textbf{262} (1985), 263--283.

\bibitem{cgm}
S. Coleman, V. Glaser and A. Martin,
Action minima among solutions to a class of euclidean scalar field equation,
\emph{Comm. Math. Phys.} \textbf{58} (1978), 211--221.

\bibitem{broglie}
L. de Broglie,
\emph{Nonlinear Wave Mechanics. A Casual Interpretation},
Elsevier Publishing Company, Amsterdam London New York Princeton, 1960.

\bibitem{garrisi}
D. Garrisi,
On the orbital stability of standing-wave solutions to a coupled nonlinear Klein-Gordon equation,
\emph{Adv. Nonlinear Stud.} \textbf{12} (2012), 639--658.

\bibitem{gss}
M. Grillakis, J. Shatah and W. Strauss,
Stability theory of solitary waves in the presence of symmetry, I, 
\emph{J. Funct. Anal.} \textbf{74} (1987), 160--197.

\bibitem{jost}
J. Jost,
\emph{Riemannian Geometry and Geometric Analysis}, 6th ed., Springer, Heidelberg, 2011.

\bibitem{mugnai}
D. Mugnai,
Solitary waves in Abelian gauge theories with strongly nonlinear potentials,
\emph{Ann. Inst. H. Poincar\'e Anal. Non Lin\'eaire} \textbf{27} (2010), 1055--1071.


\bibitem{yang}
Y. Yang,
\emph{Solitons in Field Theory and Nonlinear Analysis},
Springer-Verlag, New York, NY, 2001.

\end{thebibliography}
\end{document}